\documentclass[letterpaper, 10 pt, conference]{ieeeconf}      % Use this line for a4 paper

\IEEEoverridecommandlockouts                              % This command is only needed if 
\usepackage{geometry}
\usepackage{amsthm}      
\usepackage{amsmath,amsfonts}
\usepackage{bbm}
\usepackage{amssymb}  
\usepackage{algorithmic}
\usepackage{algorithm}
\usepackage{array} 
\usepackage{textcomp}
\usepackage{stfloats}
\usepackage{url}
\usepackage{verbatim}
\usepackage{graphicx}
\usepackage{subfigure}
\usepackage{subfloat}
\usepackage[font=small,labelfont=bf]{caption}
\usepackage[utf8]{inputenc} % Solves accented characters

\usepackage{tabularx}

\usepackage{siunitx}
\DeclareMathOperator*{\argmin}{argmin} %for argmin

\usepackage{bondgraphs}
\usepackage{tikz}
\usetikzlibrary{calc}
\usepackage{epstopdf}
\usetikzlibrary{shapes,fit}

\newtheorem{proposition}{Proposition}
\newtheorem*{objective*}{PBC objective}
\newtheorem*{tankobjective*}{Sampled PBC objective}
\newtheorem{remark}{Remark}

\usepackage{color}

\newcommand{\vect}[1]{\bm{#1}}		   % vectors
\newcommand{\matr}[1]{\bm{#1}}		   % matrices

\usepackage[style=ieee,
			url=false, % Needed for websites as bib items, but shows the 'urldate' field as 'last visited..'
			isbn=false,
			doi=true,
			giveninits=true,
			uniquename=init,
			isbn=false,
			backend=bibtex,
			minbibnames=1,
			maxbibnames=6	]{biblatex}
\bibliography{root.bib}

\title{\LARGE \bf
Lossless optimal transient control for rigid bodies in 3D space
}

\author{Riccardo Zanella$^{1}$, Federico Califano$^{1}$, Antonio Franchi$^{1}$, Stefano Stramigioli$^{1}$
\thanks{$^{1}$Robotics \& Mechatronics (RaM) group, University of Twente, The Netherlands.
       {\tt\small \{r.zanella, f.califano, a.franchi, s.stramigioli\}@utwente.nl}}%%
\thanks{This research has received funding from the European Union’s Horizon Europe Framework Programme under grant agreement No 101070596 (euROBIN).}% <-this % stops a space
}

\begin{document}

\maketitle
\thispagestyle{empty}
\pagestyle{empty}

%%%%%%%%%%%%%%%%%%%%%%%%%%%%%%%%%%%%%%%%%%%%%%%%%%%%%%%%%%%%%%%%%%%%%%%%%%%%%%%%
\begin{abstract}
In this letter, we propose a control scheme for rigid bodies designed to optimise transient behaviors. The search space for the optimal control input is parameterized to yield a passive, specifically lossless, nonlinear feedback controller. As a result, it can be combined with other stabilizing controllers without compromising the stability of the closed-loop system. The controller commands torques generating fictitious gyroscopic effects characteristics of 3D rotational rigid body motions, and as such does not inject nor extract kinetic energy from the system. We validate the controller in simulation using a model predictive control (MPC) scheme, successfully combining stability and performance in a stabilization task with obstacle avoidance constraints.
\end{abstract}

%\begin{IEEEkeywords}
%Passivity, Control Barrier Functions, Energy-aware control.
%\end{IEEEkeywords}

%%%%%%%%%%%%%%%%%%%%%%%%%%%%%%%%%%%%%%%%%%%%%%%%%%%%%%%%%%%%%%%%%%%%%%%%%%%%%%%%
\section{Introduction}

Rigid body control is of utmost importance in most control problems in the field of aerospace and robotic systems. In fact, satellites \cite{ narkiewicz2020generic, yang2023uncertain}, underwater vehicles \cite{liu2020trajectory, neira2021review}, and aerial robots \cite{su2021fast, golestani2022prescribed} are usually modeled as rigid bodies.

Over the past three decades, the concept of energy has gained significant attention in engineering practice. In a control-theoretic setting, the core idea of energy-based methods is to treat dynamical systems as physical entities able to exchange energy with other systems rather than to process signals \cite{ortega2001putting}. In this context, passivity-based control (PBC) has fully established itself as a branch of nonlinear control \cite{vanderSchaftL2}, aiming to stabilise systems independently on external environmental interactions using energy-based arguments.

The inherent passivity of mechanical systems has been extensively utilized in designing rigid body attitude controllers, as this property often allows the development of robustly stable controllers without relying on restrictive assumptions. %
As examples, in \cite{FORNI2015164}, a standard energy-balancing PBC (EB-PBC) is used for set-point tracking in a rigid body attitude control problem leveraging on port-Hamiltonian formulation. In \cite{8786163} the problem of interaction control of fully actuated unmanned aerial vehicle (UAV) is considered. There the UAV is modeled as a floating rigid body on the special Euclidean group SE(3) and an EB-PBC unified motion and impedance controller is designed. In \cite{mohammadi2021passivity} a PBC scheme is proposed for the stable co-transportation of a cable-suspended payload by a number of quadrotors.

A notable limitation of this methodology is that the low-level controller design is focusing on achieving passivity of the closed-loop system, without optimizing task performance. This aspect produced the highly debated assertion that passivity is associated to a loss of performance in the control due to the choice of over-conservative gains \cite{dimeas2016online,Califano2022OnSystems}. 

Recently, researchers have begun to investigate how to integrate optimization techniques into passive controllers, aiming to preserve passivity (and as a consequence stability) of the closed-loop system while also maximizing performance.
For example, in \cite{raff2007nonlinear}, a nonlinear model predictive control (MPC) scheme has been combined for the first time with the passivity concept to merge performance characteristic of MPC with closed-loop stability guarantees. 
For example, \cite{raff2007nonlinear} was the first to combine a nonlinear model predictive control (MPC) scheme with the passivity concept, effectively merging the performance benefits of MPC with the stability guarantees of passive control. This method guarantees closed-loop stability by incorporating a passivity constraint directly into online optimization.
Conversely, a quasi-Linear-Parameter-Varying MPC framework from \cite{cisneros2018dissipativity} removes the need of extra constraints or terminal costs in optimization by ensuring stability with dissipativity arguments, given certain nonlinearity conditions in constrained MPC are satisfied.
More recently, other works propose to carefully integrate the expressive power of neural networks into PBC schemes \cite{zanella2024learning, massaroli2022optimal, pmlr-v168-plaza22a} by systematically designing parametrised passive controllers through the solution of an optimization problem.
% %
Other works combine energy-based schemes and optimisation tools in a model-free fashion, where the PBC design involving energy tanks are merged in a reinforcement learning framework \cite{zanella2024learningpassive,sacerdoti2024reinforcement}.

In this work, we propose a new approach that integrates optimization with PBC methods to optimize transient behavior while ensuring passivity. We focus on controlled rigid bodies where a stabilizing controller is employed, and present a novel parameterization of an extra passive control input, able to shape the transient response while preserving stability. This approach strikes a balance between the expressiveness of the parameterization (capable of accommodating diverse dynamic responses) and the requirement to preserve closed-loop stability.
The main contributions of this work can be summarized as follows:
%In this work, we share the motivation of combining optimization tools with a PBC scheme and propose an original perspective to achieve this goal. The main contributions of this work can be summarized as follows:
\begin{itemize}
    \item[-] We introduce a novel framework that parametrizes the search space for the optimal control input, ensuring a structurally passive nonlinear feedback controller.
    \item[-] We show how the net effect of the proposed framework is to produce fictitious gyroscopic effects, effectively exploiting the nonlinear 3D dynamics of rigid body systems during transient phases.
    \item[-] We discuss how the proposed control methodology integrates with other stabilizing controllers and optimization-based techniques to shape transient behaviours while preserving closed-loop stability.
    \item[-] We validate the method through simulations, optimizing transient behaviour with an MPC scheme for obstacle avoidance, while a stabilizing nominal controller drives the rigid body to a target.  
\end{itemize}

The paper is organised as follows. Section \ref{sec:background} contains the background material comprising passivity theory and rigid body modeling. Section \ref{sec:main} contains the theoretical contribution of this work in which the proposed optimal control scheme is presented. Section \ref{sec:sim} contains the description, simulation results and discussion of the proposed case study, while Section \ref{sec:conc} contains conclusions and future work. 

% \begin{figure}[t]
%     \centering
%     \includegraphics[width=0.5\linewidth]{figures/rigidbody.pdf}
%     \caption{Frame description of floating rigid body.}
%     \label{fig:rb}
% \end{figure}
%%%%%%%%%%%%%%%%%%%%%%%%%%%%%%%%%%%%%%%%%%%%%%%%%%%%%%%%%%%%%%%%%%%%%%%%%%%%%%%%%%%%%%%%%%%%%%%%%%%%%%%%%%%%%%%%%%%%%%%%%%%%%%%%%%%%%%%%%%%%%%%%%%%%%%%%%%%%%%%%%%%%%%%%%%%%%%%%%%
%%%%%%%%%%%%%%%%%%%%%%%%%%%%%%%%%%%%%%%%%%%%%%%%%%%%%%%%%%%%%%%%%%%%%%%%%%%%%%%%%%%%%%%%%
\section{Background}
\label{sec:background}

In this section we briefly report the relevant background information about passivity and rigid body modelling, which can be found in detail e.g., in \cite{vanderSchaftL2}.
Throughout the paper, scalars and maps are denoted by regular letters, while vectors and matrices are denoted by bold letters. The skew-symmetric representation of a three-dimensional vector $\vect{\omega} \in \mathbb{R}^3$ is denoted $\tilde{\vect{\omega}}=-\tilde{\vect{\omega}}^{\top} \in \mathbb{R}^{3 \times 3}$. Vector products of two vectors $\vect{a}\in \mathbb{R}^3$ and $\vect{b}\in \mathbb{R}^3$ is denoted by $\vect{a} \times \vect{b}$ and we will use the property $\vect{a} \times \vect{b}=\tilde{\vect{a}} \vect{b}$.

\subsection{Passivity}

Consider the affine nonlinear control system:
\begin{equation}
\label{eq:nonlinearaffinesystem}
    \dot{\vect{x}}=f(\vect{x})+g(\vect{x})\vect{u}
\end{equation}
where $\vect{x}\in \mathcal{D} \subseteq \mathbb{R}^n$ is the state, $\vect{u}\in \mathcal{U} \subset \mathbb{R}^m$ is the input, $f: \mathbb{R}^{n} \to \mathbb{R}^{n}$ and $g: \mathbb{R}^{n} \to \mathbb{R}^{n \times m}$ are continuously differentiable maps. As a consequence, a Lipschitz continuous controller guarantees the existence and uniqueness of solutions of (\ref{eq:nonlinearaffinesystem}). 
A system in the form (\ref{eq:nonlinearaffinesystem}) equipped with an output: \[\vect{y}=h(\vect{x})\in \mathcal{Y} \subset \mathbb{R}^m\] is said to be \textit{passive} with respect to a differentiable \textit{storage function} $S : \mathcal{D} \to \mathbb{R}^+$ and input-output pair $(\vect{u},\vect{y})$, if the following inequality holds $\forall \vect{u} \in \mathcal{U}$:

\begin{equation}
\label{eq:passivity}
 \dot{S}=\partial_{\vect{x}}^{\top} S(\vect{x}) (f(\vect{x})+g(\vect{x})\vect{u}) \leq \vect{y}^{\top}\vect{u}.
\end{equation}

For physical systems, where $S(\vect{x})$ represents the energy and $\vect{y}^{\top}\vect{u}$ the power flowing in the system, condition (\ref{eq:passivity}) is a statement of energy conservation, i.e., the variation of energy in the system is bounded by the power flowing in the system. The inequality margin in (\ref{eq:passivity}) is due to the \textit{dissipation} of the system $d(\vect{x})=-\partial_{\vect{x}}^{\top} S(\vect{x}) f(\vect{x})$. For passive systems, it holds $d(\vect{x})\geq 0$ and the case $d(\vect{x})=0$ (i.e., (\ref{eq:passivity}) with an equal sign), system (\ref{eq:nonlinearaffinesystem}) is said to be \textit{lossless}.
Passivity is directly related to local stability of the minimum of $S(\vect{x})$, under the condition that the storage function qualifies as a Lyapunov function. In fact, it is easy to see that in this case passivity condition (\ref{eq:passivity}) with $\vect{u}=0$ makes the minimum of the storage/Lyapunov function $S(\vect{x})$ locally stable as $\dot{S}=-d(\vect{x})\leq 0$. The rate of convergence to the equilibrium depends on the value of the dissipation $d(t)$ and, as limit case, if the system is lossless, the system evolves on isolines of $S(\vect{x})$. Notably, passivity yields a stronger property than just stability: passivity can be shown to be necessary for a robust form of stability when external interactions influence the system through the input-output pair $(\vect{u},\vect{y})$, see \cite{Califano2022OnSystems,Capelli2022PassivityEnergy,Stramigioli2015Energy-AwareRobotics} for more details.

This energetic analysis can be extended to closed-loop systems and, as such, used for control design. The so-called \textit{passivity-based control} (PBC) methods aim to design controllers such that the closed-loop system is passive, providing an energetic interpretation of the controlled system and a form of robust stability with respect external and parametric disturbances.

\subsection{Rigid Body Model}

Consider a rigid body evolving in Euclidean 3D space with mass $m$ and inertia tensor $\matr{J}$.
% With reference to Fig. \ref{fig:rb}, let $\Psi_0$ and $\Psi_b$ be the inertial and body-fixed frame, where the latter is chosen to be the \textit{principal inertial frame}, which is a specific body-fixed frame with origin coincident with the centre of mass of the rigid body and in which its inertia tensor $\matr{J}\in \mathbb{R}^{3 \times 3}$ is a diagonal matrix. 
Let $\Psi_0$ and $\Psi_b$ be the inertial and body-fixed frame, where the latter is chosen to be the \textit{principal inertial frame}, which is a particular body-fixed frame with origin coincident with the centre of mass of the rigid body and in which its inertia tensor $\matr{J}\in \mathbb{R}^{3 \times 3}$ is a diagonal matrix. 
The rigid body configuration is encoded in the Lie group element $\matr{H} \in SE(3)$ connecting $\Psi_0$ to $\Psi_b$, which contains, in its usual matrix form, the rotation matrix $\matr{R}\in SO(3)$ and the distance $\matr{p}\in \mathbb{R}^3$.
%
% \begin{equation}
%     \matr{H}=\begin{pmatrix}
%     \matr{R} & \vect{p} \\ 0 \, 0 \, 0 & 1
%     \end{pmatrix}.
% \end{equation}
%
The kinematic variables are completed by the generalised velocity of the rigid body, composed by the linear velocity $\vect{v}=\dot{\vect{p}}\in \mathbb{R}^3$, and the angular velocity $\vect{\omega}\in \mathbb{R}^3$ expressed in $\Psi_b$. The relation between $\vect{\omega}$ and $\vect{R}$ is encoded in the kinematic identity 
% $\tilde{\vect{\omega}}=\vect{R}^{\top}\dot{\vect{R}}$,
\begin{equation}
    \label{eq:omegatilde}
\tilde{\vect{\omega}}=\vect{R}^{\top}\dot{\vect{R}},
\end{equation}
where the Lie algebra element $\tilde{\vect{\omega}}\in so(3)$ is represented by a skew-symmetric matrix and is isomorphic to $\vect{\omega}=(\omega_x,\omega_y,\omega_z)^{\top}$ through: 
\begin{equation}
    \tilde{\vect{\omega}} = \begin{bmatrix}
    0 & -\omega_z & \omega_y \\
    \omega_z & 0 & -\omega_x \\
    -\omega_y & \omega_x & 0
\end{bmatrix}.
\end{equation}

A state-space representation of frictionless equations of a rigid body are:
\begin{align}
\label{eq:kin1}
    \dot{\vect{p}}&=\vect{v} \\
    \label{eq:kin2}
    \dot{\matr{R}}&=\matr{R} \tilde{\vect{\omega}} \\
    \label{eq:dyn1}
    m\dot{\vect{v}} &=- m g\vect{e}_3 + \matr{R} \vect{f} \\
    \label{eq:dyn2}
    \matr{J} \dot{\vect{\omega}} &=- \vect{\omega} \times \matr{J} \vect{\omega} + \vect{\tau},
\end{align}
where $\vect{e}_3$ is the direction of gravity, and $\matr{f}\in \mathbb{R}^3$ and $\matr{\tau}\in \mathbb{R}^3$ are the force and momentum applied at the center of mass of the rigid body and expressed in $\Psi_b$. Notice that the first two equations are kinematic identities, while the last two are the dynamic equations.

% \todo{
% \begin{remark}
%     When system (\ref{eq:kin1}-\ref{eq:kin2}-\ref{eq:dyn1}-\ref{eq:dyn2}) is used to model multi-rotor aerial vehicles with $n$ rotors, the actual control input is related to the rigid body wrench through an allocation matrix $\matr{G} \in \mathbb{R}^{6 \times n}$ by
%     \begin{equation}
%         \begin{pmatrix}
%             \vect{f} \\ \vect{\tau}
%         \end{pmatrix}=\matr{G} \vect{\lambda},
%     \end{equation}
%     where $\vect{\lambda}\in \mathbb{R}^{n}$ collects the intensity of the forces produced by each rotor. The rank of $\matr{G}$ determines then the actuation degree of the system, which is fully actuated in case $\textrm{rank}(\matr{G})=6$.
% \end{remark} 
% }

It is easy to see that, after the application of a gravity compensation term in the form:
\begin{equation}
\label{eq:gravitycomp}
     \matr{R}\vect{f}=m g \vect{e}_3+ \vect{f}',
\end{equation}
system (\ref{eq:dyn1}-\ref{eq:dyn2}) is lossless with respect to the kinetic energy 
\begin{equation}
    K(\vect{\omega},\vect{v})=\frac{1}{2}m\vect{v}^{\top} \vect{v} + \frac{1}{2}\vect{\omega}^{\top} \matr{J} \vect{\omega}
\end{equation}
as a storage function, and input/output pairs $(\vect{f}', \vect{v})$ and $(\vect{\tau} ,\vect{\omega})$. In fact, calculating the dissipation inequality and using (\ref{eq:gravitycomp}), one obtains:
\begin{equation}
\label{eq:dissInequality}
    \dot{K}=\vect{v}^{\top}\vect{f}'+\vect{\omega}^{\top}(\vect{\tau}-\vect{\omega} \times \matr{J} \vect{\omega})=\vect{v}^{\top}\vect{f}'+\vect{\omega}^{\top}\vect{\tau}.
    \end{equation}
The latter follows from the key identity 
\begin{equation}
\label{eq:tripleproduct}
     \vect{\omega}^{\top}(\vect{\omega} \times \matr{a})=\vect{\omega}^{\top}(\tilde{\vect{\omega}} \matr{a})=-(\tilde{\vect{\omega}}\vect{\omega})^{\top} \matr{a}=-({\vect{\omega} \times \vect{\omega}})\vect{a}=0,
\end{equation}
which holds true for any $\matr{a}\in \mathbb{R}^{3}$. 
\section{Lossless control and Optimisation}\label{sec:main}

The following results motivates the choice of controller proposed in the sequel.
\begin{proposition}
\label{prop:1}
    Consider system (\ref{eq:dyn1}-\ref{eq:dyn2}) with control (\ref{eq:gravitycomp}) and any torque controller in the form
    \begin{equation}
        \label{eq:tau}
        \vect{\tau}=\vect{\omega} \times \vect{a}+\vect{\tau}'
    \end{equation}for any (possibly state-dependent) term $\vect{a}\in\mathbb{R}^3$.  Then,  the closed-loop system is passive with respect to its kinetic energy as storage function and input/output pairs $(\vect{f}', \vect{v})$ and $(\vect{\tau}' ,\vect{\omega})$.  
\end{proposition}
\begin{proof}
The closed-loop system results in:
\begin{align}
    \label{eq:1kin1}
    \dot{\vect{p}}&=\vect{v} \\
    \label{eq:1kin2}
    \dot{\matr{R}}&=\matr{R} \tilde{\vect{\omega}} \\
    \label{eq:1dyn1}
    m\dot{\vect{v}} &= \vect{f}' \\
    \label{eq:1dyn2}
    \matr{J} \dot{\vect{\omega}} &=- \vect{\omega} \times (\matr{J} \vect{\omega}-\vect{a}) + \vect{\tau}'.
\end{align} 
Passivity follows by the same calculation as in (\ref{eq:dissInequality}) with the extra zero addend (\ref{eq:tripleproduct}), resulting in 
\begin{equation}
\label{eq:dissipativityproof}
    \dot{K}=\vect{v}^{\top}\vect{f}'+\vect{\omega}^{\top}(\vect{\tau}'-\vect{\omega} \times (\matr{J} \vect{\omega}-\vect{a}))=\vect{v}^{\top}\vect{f}'+\vect{\omega}^{\top}\vect{\tau}'.
\end{equation} 
\end{proof}

\begin{proposition}
\label{prop:2}
    Consider the system (\ref{eq:kin1}-\ref{eq:kin2}-\ref{eq:dyn1}-\ref{eq:dyn2}) with control (\ref{eq:gravitycomp}) and (\ref{eq:tau}). Suppose that the controllers commanding $\vect{f}'$ and $\vect{\tau}'$ are designed to stabilize a minimum of a Lyapunov function of the form $V(\matr{R},\vect{p},\vect{\omega},\vect{v})=K(\vect{\omega},\vect{v})+U(\matr{R},\vect{p})$, where $K(\vect{\omega},\vect{v})$ is the kinetic energy  and $U(\matr{R},\vect{p})$ is a closed-loop potential function. Then, the closed-form expression of $\dot{V}$ is invariant under any choice of $\vect{a}\in\mathbb{R}^3$.
\end{proposition}

\begin{proof} 
To prove invariance of the expression of $\dot{V}$ under the term $\vect{\omega} \times \vect{a}$ and any controller $\vect{f}'$ and $\vect{\tau}'$, it suffices to observe that the control term $\vect{\omega} \times \vect{a}$ enters only in $\dot{K}$, producing a null additive contribution as in (\ref{eq:dissipativityproof}).
\end{proof}
The previous result can be notably interpreted from an energetic perspective. In fact, a control torque parameterised as $\vect{\tau}=\vect{\omega} \times \vect{a}$ purely \textit{routes} mechanical energy in the system, redistributing the components of the angular velocity of the rigid body without contributing to a change of kinetic energy. In this sense, the mechanism is analogous to what 3D \textit{gyroscopic terms} do, which are in fact an instance of $\vect{\omega} \times \vect{a}$ with $\vect{a}=-\matr{J}\vect{\omega}$, as clear in (\ref{eq:dyn2}).

\begin{remark}[Closed-loop inertia shaping]
Although it might be tempting to think that the fictitious gyroscopic effects $\vect{\vect{\omega}} \times \vect{a}$ would effectively shape the inertia of the rigid body, simple calculations show that the rotational closed-loop dynamics
\begin{equation}
    \label{eq:closed-loop}
    \matr{J} \dot{\vect{\omega}} = -\vect{\omega} \times (-\vect{a}+\matr{J} \vect{\omega})+\vect{\tau}'
\end{equation}
is not equivalent to the dynamics of a rigid body with desired inertia $\vect{J}_d$
\begin{equation}
\label{eq:rigiddesired}
    \matr{J}_d \dot{\vect{\omega}} = - \vect{\omega} \times \matr{J}_d \vect{\omega}+\vect{\tau}',
\end{equation}
unless
\begin{equation}
    \tilde{\vect{\omega}}\vect{a}= \tilde{\vect{\omega}}\matr{J}\vect{\omega}-\matr{J}\matr{J}_d^{-1}\tilde{\vect{\omega}}\matr{J}_d\vect{\omega}+(\matr{J}\matr{J}_d^{-1}-\matr{I}_3)\vect{\tau}'
\end{equation}
holds true. Notice that this matching condition is in general not solvable for $\vect{a}$, and also the case which is normally referred to as ``feedback linearisation" $\vect{a}=\vect{J}\vect{\omega}$ does not generate any physically plausible closed-loop system (\ref{eq:rigiddesired}), unless $\vect{J}_d=\vect{J}=\vect{I}_3$, i.e., the rigid body is dynamically equivalent to a sphere and the controller does nothing.
\end{remark}

From a stabilization perspective, Proposition \ref{prop:2} produces a degree of freedom in the control design: If the controllers commanding $\vect{f}'$ and $\vect{\tau}'$ have been designed to achieve closed-loop stability encoded in the Lyapunov function $V(\matr{R},\vect{p},\vect{\omega},\vect{v})=K(\vect{\omega},\vect{v})+U(\matr{R},\vect{p})$, $V$ would still be a closed-loop Lyapunov function for an additional control term in the form $\vect{\omega} \times \vect{a}$. As a consequence, this control term can be chosen to achieve desired \textit{transient behaviors} of rigid bodies, complementing stabilising controllers with extra capabilities functional to specific task executions without compromising stability.

With this observation in mind, we propose the following abstract optimal control problem (OCP) in the time interval $t\in [t_i,t_f]$:
\begin{equation}
\label{eq:LQ}
\begin{aligned}
\vect{a}^*(t)= & \argmin_{\vect{a}(t)\in \mathcal{A}} \quad  C(\vect{R},\vect{p},\vect{\omega},\vect{v},\vect{\tau}',\vect{f}',\vect{a},t_i,t_f)\\
 & \textrm{s.t.}  \quad  \quad \textrm{(Dynamics (\ref{eq:1kin1}-\ref{eq:1kin2}-\ref{eq:1dyn1}-\ref{eq:1dyn2}) holds})\\
 & \quad \quad \quad \, \textrm{(Actuation constraints)}\\
 & \quad \quad \quad \,
 \textrm{(Task-based constraints)}.
 \end{aligned}
\end{equation}
Here $C$ represents a general cost function, which can be composed of different addends like control effort, regularizing terms, and soft task-based constraints. The hard constraints are composed by the dynamics of the system (\ref{eq:1kin1}-\ref{eq:1kin2}-\ref{eq:1dyn1}-\ref{eq:1dyn2}), the actuation constraints (e.g., limits in maximal torques or velocities), and/or hard task constraints (e.g., obstacle avoidance). 

The specific technique to tackle the OCP (\ref{eq:LQ}), which comprises the nature of the search space $\mathcal{A}$, is driven by the specific application. The goal of presenting the OCP in this abstract form is to highlight its generality, i.e., whatever technique is used to actually solve the OCP, the extra torque term $\vect{\tau}=\vect{a} \times \vect{\omega}$ will not interfere with the stability properties of the closed-loop system.  Section \ref{sec:sim} will provide a detailed example of mathematical representation for the constraints and the cost function.

% \section{Constrained optimisation}
% \label{sec:4}

%%%%%%%%%%%%%%%%%%%%%%%%%%%%%%%%%%%%%%%%%%%%%%%%%%%%%%%%%%%%%%%%%%%%%%%%%%%%%%%%%%%%%%%%%%%%%%%%%%%%%%%%%%%%%%%%%%%%%%%%%%%%%%%%%%%%%%%%%%%%%%%%%%%%%%%%%%%%%%%%%%%%%%%%%%%%%%%%%%%%%%%%%%%%%%%%%%%%%%%%%%%%%%%%%%%%%%%%%%%%%%%%%%%%%%%%%%%%%%%%%%%%%%%%%%%%%%%%%%%%%%%%%%%
\section{Case studies and Simulations}\label{sec:sim}

\subsubsection{Discrete model and software setup}
In the following, we will present a \textit{non-linear model predictive control} (NMPC) solution to the problem. This choice is motivated by the capability of MPC schemes to integrate generic constraints in an optimal control framework, making it a good candidate to implement the abstract problem (\ref{eq:LQ}).
% This choice is motivated by the ability of MPC schemes to accommodate a diverse range of criteria while simultaneously ensuring smooth integration with the inherent stability properties of the proposed approach. %We will show how in an NMPC framework the abstract constraints in the OCP can be conveniently handled. 
%
The NMPC uses a temporal discretisation of the system and solves the OCP (\ref{eq:LQ}) at each time step in the prediction horizon $[t_i,t_f]$. Within this framework, the equations (\ref{eq:kin1}), (\ref{eq:dyn1}) are integrated numerically using the  Euler's method with sample time $T$:
\begin{align} 
    \vect{p}_{k+1} &= \vect{p}_k + T\vect{v}_k \label{eq:dt_kin1}\\
    m\vect{v}_{k+1} & = \vect{v}_k - m g T\vect{e}_3 + T\matr{R}_k\vect{f}_k \label{eq:dt_dyn1}.
\end{align} 
Meanwhile, the explicit Lie-Newmark method \cite{simo1988dynamics}, is applied to (\ref{eq:kin2}), (\ref{eq:dyn2}) to ensure that the predicted $\matr{R}_{k+1}$ remains in $S O(3)$. Given the current configuration at the $k$-th instant of time $\left(\matr{R}_k, \vect{\omega}_k\right)$, this method updates $(\matr{R}_{k+1}, \vect{\omega}_{k+1})$ using the following iteration rule:
\begin{align} 
    \matr{R}_{k+1} & =\matr{R}_k \operatorname{cay}\left(T \vect{\omega}_{k+\frac{1}{2}}\right) \label{eq:dt_kin2}\\
    \vect{\omega}_{k+1} & =\vect{\omega}_{k+\frac{1}{2}}+\frac{T}{2} \matr{J}^{-1}\left(\matr{J}  \vect{\omega}_{k+\frac{1}{2}} \times \vect{\omega}_{k+\frac{1}{2}}+\vect{\tau}_k\right) \label{eq:dt_dyn2}
\end{align} 
where $\vect{\omega}_{k+\frac{1}{2}} =\vect{\omega}_k+\frac{T}{2} \matr{J} ^{-1}\left(\matr{J} \vect{\omega}_k \times \vect{\omega}_k+\vect{\tau}_k\right)$ and $\operatorname{cay}(\cdot)$ is the Cayley map defined by
$
\operatorname{cay}(x)=\left(I_3+\frac{1}{2} \tilde{x}\right)\left(I_3-\frac{1}{2} \tilde{x}\right)^{-1}
$, 
with $I_3$ the identity matrix. 

 Employing this approximation alongside a discrete-time cost function specially formulated for the task, the OCP (\ref{eq:LQ}) transforms into a finite-dimensional nonlinear program (NLP). %\todo{After this approximation, a direct multiple-shooting technique \cite{bock1984multiple} is used to solve (\ref{eq:LQ}). These techniques convert a continuous-time optimization problem of the form (\ref{eq:LQ}) into a finite-dimensional OCP by finitely parameterizing the infinite-dimensional decision variables, thus approximating the original problem with a finite-dimensional nonlinear program (NLP). }
The NLP in our simulations is solved by CasADI \cite{andersson2019casadi}, a toolbox for Algorithmic Differentiation, interfaced with IPOPT \cite{wachter2002interior} solver.

%%%%%%%%%%%%%%%%%%%%%%%%%%% 
\subsubsection{Nominal controller}
To validate the proposed control scheme, we consider a rigid body controlled by (\ref{eq:gravitycomp}) and (\ref{eq:tau}), where $\vect{f}' =\vect{0}$, while $ \vect{\tau}'$ is designed to ensure the stability of the closed-loop system around the desired orientational equilibrium $\matr{R}=\matr{R}_d\in SO(3)$. For this purpose we use the controller developed in \cite{rashad2019port}: 
\begin{equation} \label{eq:rot_ctrl}
    \vect{\tau}'=-2\operatorname{sk}(\matr{G} \matr{R}_d^\top \matr{R}) - k_D \vect{\omega},
\end{equation}  
where $\operatorname{sk}(A)$ denotes the anti-symmetric part of a matrix $A$, given by $\operatorname{sk}(A)=\frac{A-A^\top}{2}$, while $\matr{G}\in\mathbb{R}^{3\times3}$ and $k_D\in\mathbb{R}$ are the symmetric orientational
co-stiffness matrix and the positive damping gain, respectively. 
As shown in \cite{rashad2019port}, the controller (\ref{eq:rot_ctrl}) makes the closed-loop system almost globally asymptotically stable around the desired orientational equilibrium $\matr{R}=\matr{R}_d\in SO(3)$, as can be checked using the Lyapunov candidate  
\begin{equation}\label{eq:Lyapunov}
V(\matr{R},\vect{p},\vect{\omega},\vect{v})=K(\vect{\omega},\vect{v})+U(\matr{R},\vect{p})
\end{equation}
with $ K(\vect{\omega},\vect{v})= \frac{1}{2}m\vect{v}^\top \vect{v}+\frac{1}{2}\vect{\omega}^\top \matr{J}\vect{\omega}$ and $ U(\matr{R},\vect{p}) = -\operatorname{tr}(\matr{G}(\matr{R}_d^\top \matr{R}-\matr{I}_3))$, where $ \operatorname{tr}(\cdot)$ is the matrix trace.

%%%%%%%%%%%%%%%%%%%%%%%%%%
\subsubsection{Task description}
As a use case scenario, we consider a rigid body moving at a constant linear velocity ($\vect{f}'=\vect{0}$) toward a wall with a slit. The rigid body is thinner along one spatial dimension, and as such the object must align its predominant axes with the slit to avoid collision. We introduce a constraint to encode both the orientation error $\varepsilon_{R,R_\star}$ and the distance $d_{p,p_\star}$ between the rigid body and the slit:
\begin{equation} \label{eq:constraint}
    \frac{\varepsilon_{R,R_\star}}{ d_{p,p_\star}+\epsilon_1}<\epsilon_2
\end{equation}
where we defined 
\begin{align} 
\varepsilon_{R,R_\star} :=& 1-(\vect{e}_3^\top R^\top R_\star \vect{e}_3)^2\\ 
d_{p,p_\star} :=& \left(\vect{e}_1^\top (\vect{p}-\vect{p}_\star)\right)^2 + \left(\vect{e}_2^\top (\vect{p}-\vect{p}_\star)\right)^2
\end{align}
with 
$\vect{e}_1=[1 \ 0 \ 0]^{\top},\vect{e}_2=[0 \ 1 \ 0]^{\top},\vect{e}_3=[0 \ 0 \ 1]^{\top}$, while $\vect{p}_\star\in \mathbb{R}^3$ and $\matr{R}_\star\in SO(3)$  are respectively the position and orientation of the slit with respect to the inertial frame.  The design of this constraint reflects that as the rigid body approaches the slit (smaller $d_{p,p_\star}$), misalignment becomes increasingly critical. The parameters $\epsilon_1 \in \mathbb{R}^+$ prevents numerical instability or overly aggressive control responses, whereas $\epsilon_2 \in \mathbb{R}^+
$ limits the orientation error relative to distance. Reducing $\epsilon_2$ increases sensitivity to alignment errors as the body approaches the slit.

\subsubsection{Simulation with nonlinear MPC}
The system is simulated for $50 \si{s}$ by discretising the system dynamics (\ref{eq:dt_kin1})-(\ref{eq:dt_dyn2}) for $25000$ steps. The rigid body has $m=1 \si{kg}$ and $\matr{J}=\operatorname{diag}(1.5,1.5,0.5) \si{kg\ m^2}$, with the following initial conditions: $\vect{p}_0 = \left[ 0 \ \ 0\ \ 0.5\right]\si{m}$, $\vect{v}_0 = \left[ 0 \ \ 0.1\ \ 0\right]\si{m/s}$, $\vect{\omega}_0 = \left[ 0.15 \ \ -0.2\ \ 0.25\right]\si{rad/s}$, $\matr{R}_0 = \matr{I}_3$. The slit is in position $\vect{p}_\star = \left[ 0 \ \ 2.5\ \ 0\right]\si{m}$ and has an orientation $\matr{R}_\star=I_3$.
At each time step, the NMPC solves the following OCP at discrete time in the prediction horizon $k=0, \dots, H-1$:
\begin{align*}
\label{eq:MPC} 
    \vect{a}_0^*,\dots,\vect{a}^*_{H-1}= & \argmin_{\vect{a}_0,\dots,\vect{a}_{H-1} \in \mathbb{R}^{3}} \quad  \sum_{k=0}^{H-1}\|\ \tilde{\vect{\omega}}_k \vect{a}_k  \|^2  \\
     & \textrm{s.t.}  \quad  \quad (\ref{eq:dt_kin1})-(\ref{eq:dt_dyn2}), (\ref{eq:constraint})\\
     & \quad \quad \quad \, |a^{i}_k|<\Bar{a}^{i}, i=1,2,3  
\end{align*}
where $a_k^i \in\mathbb{R}$ are the components of $\vect{a}_k\in\mathbb{R}^3$ and $\Bar{a}^i\in\mathbb{R}$ are the corresponding bounds.  This serves as an explicit formulation example for the OCP outlined in equation (\ref{eq:LQ}), specifically designed for the given problem. Since we implement an MPC scheme involving temporal discretization of the system, rather than optimizing a continuous-time trajectory \(a(t)\) for \(t \in [t_i, t_f]\), a sequence of discrete-time terms is determined: \(a_k := a(k \cdot T + t_i)\) for \(k = 0, \ldots, H-1\), where \(H = \frac{t_f - t_i}{T}\) and \(T\) represents the discretization interval.  
The prediction horizon of the OCP in our simulations is 1$\si{s}$ (implemented with sampling time  $T=0.1 \si{s}$ and $H=10$), while the constraint is characterised by $\epsilon_1 = 0.005, \epsilon_2 = 0.1$. The resulting control input, updated at a frequency of 500 $\si{Hz}$, takes the form (\ref{eq:tau}) with $\vect{a}=\vect{a}^*_0$.

\begin{figure*}[t]  % * makes the figure span across two columns
    \centering
    % Left column figure (Free dynamics)
    \begin{minipage}[t]{0.48\textwidth}  % Adjust width for left figure
        \centering
        \subfigure{
            \includegraphics[width=0.5\textwidth]{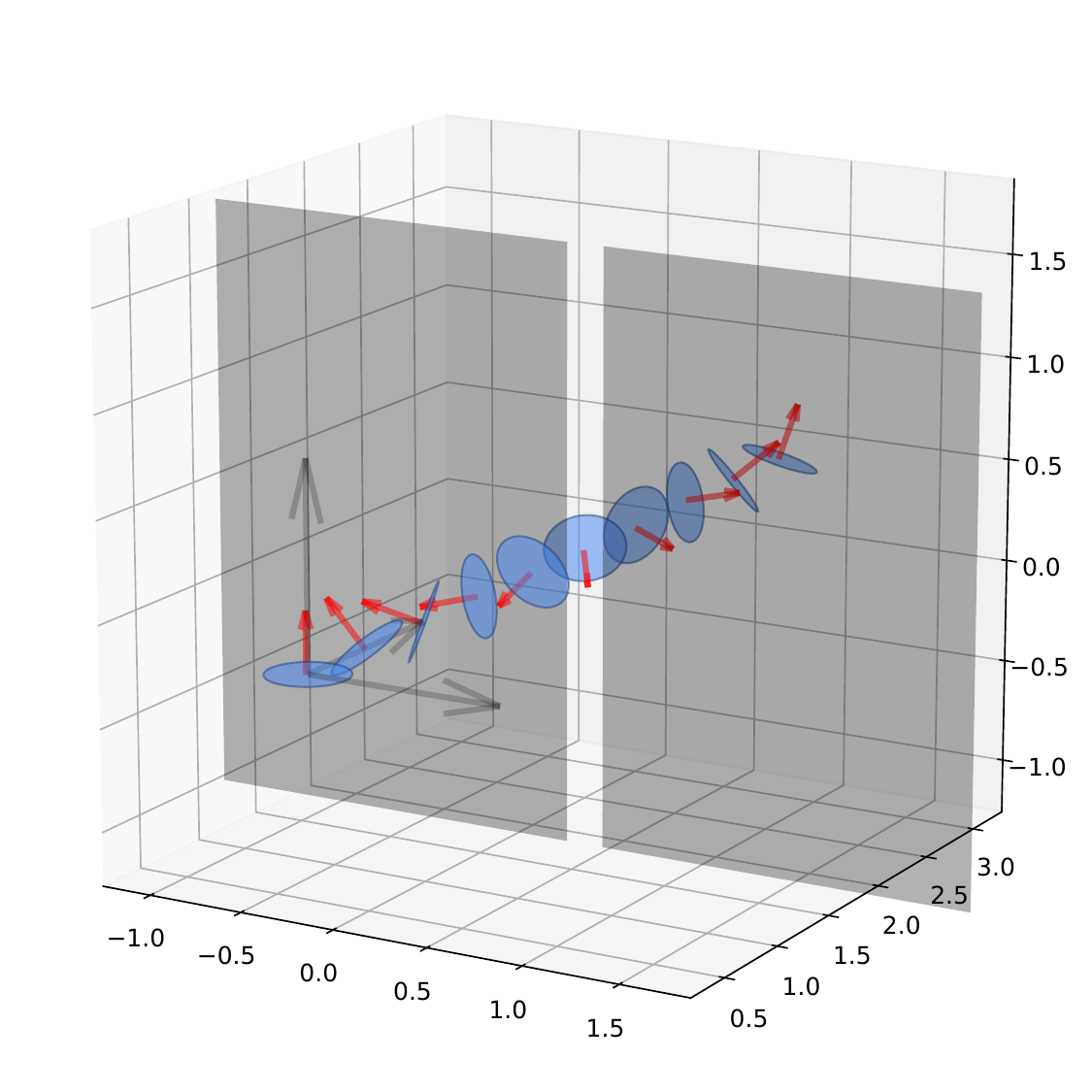}
            \label{fig:free/3d}
        } \hspace{-7mm} 
        \subfigure{
            \includegraphics[width=0.46\textwidth]{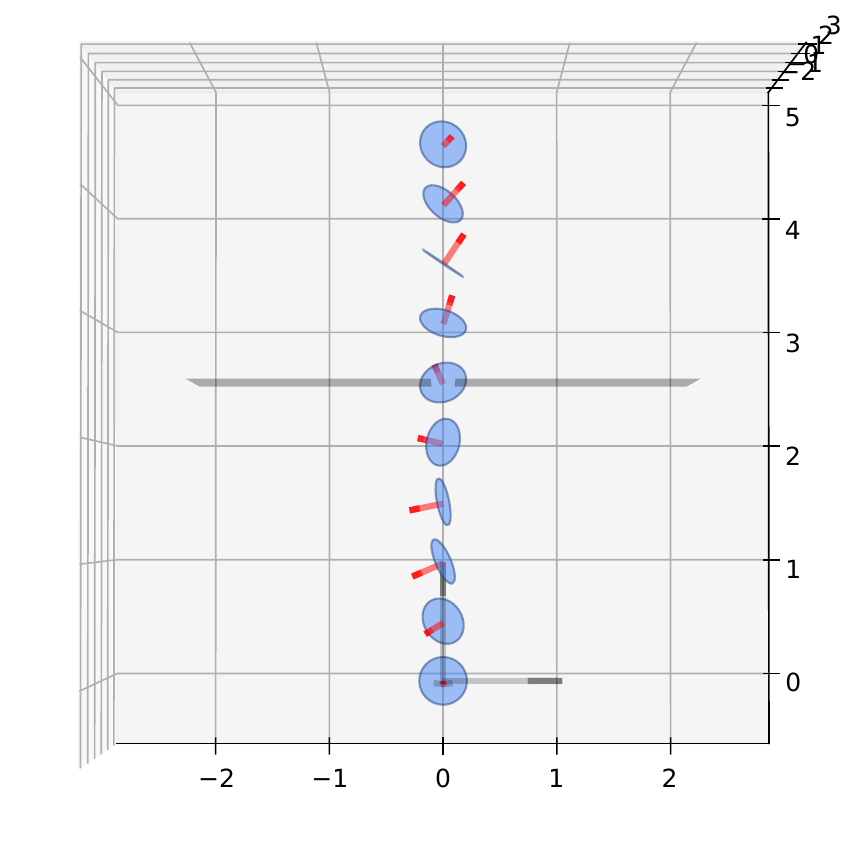}
            \label{fig:free/3dtop}
        } \\ \vspace{-2mm} 
        \subfigure{
            \includegraphics[width=0.9\textwidth]{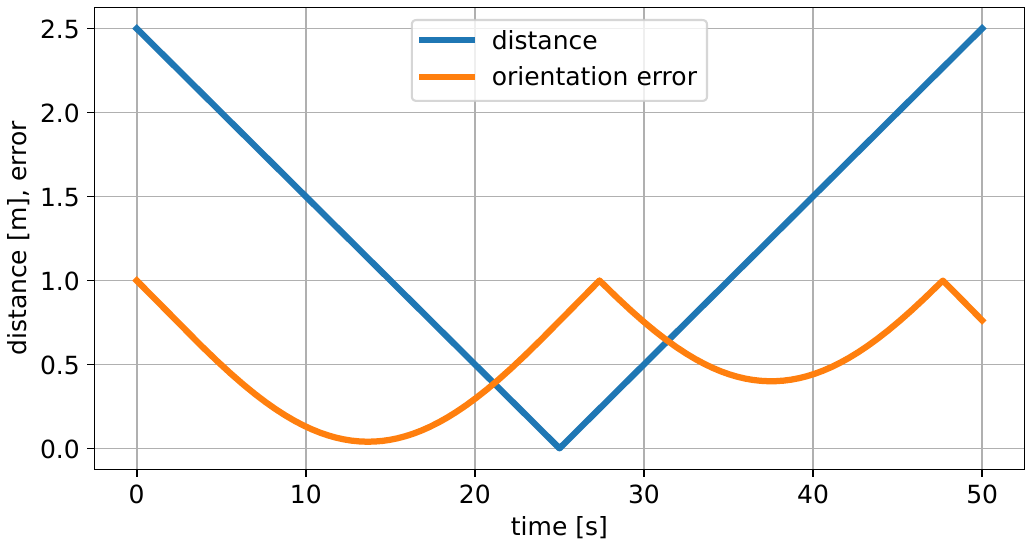}
            \label{fig:free/errors}
        }
        \caption{Free dynamics ($\vect{\tau}'=\vect{a}=\vect{0}$)}
        \label{fig:free}
    \end{minipage}%
    \hspace{5mm}  % Space between the two columns
    % Right column figure (Lossless optimal control without torque control)
    \begin{minipage}[t]{0.48\textwidth}  % Adjust width for right figure
        \centering
        \subfigure{
            \includegraphics[width=0.5\textwidth]{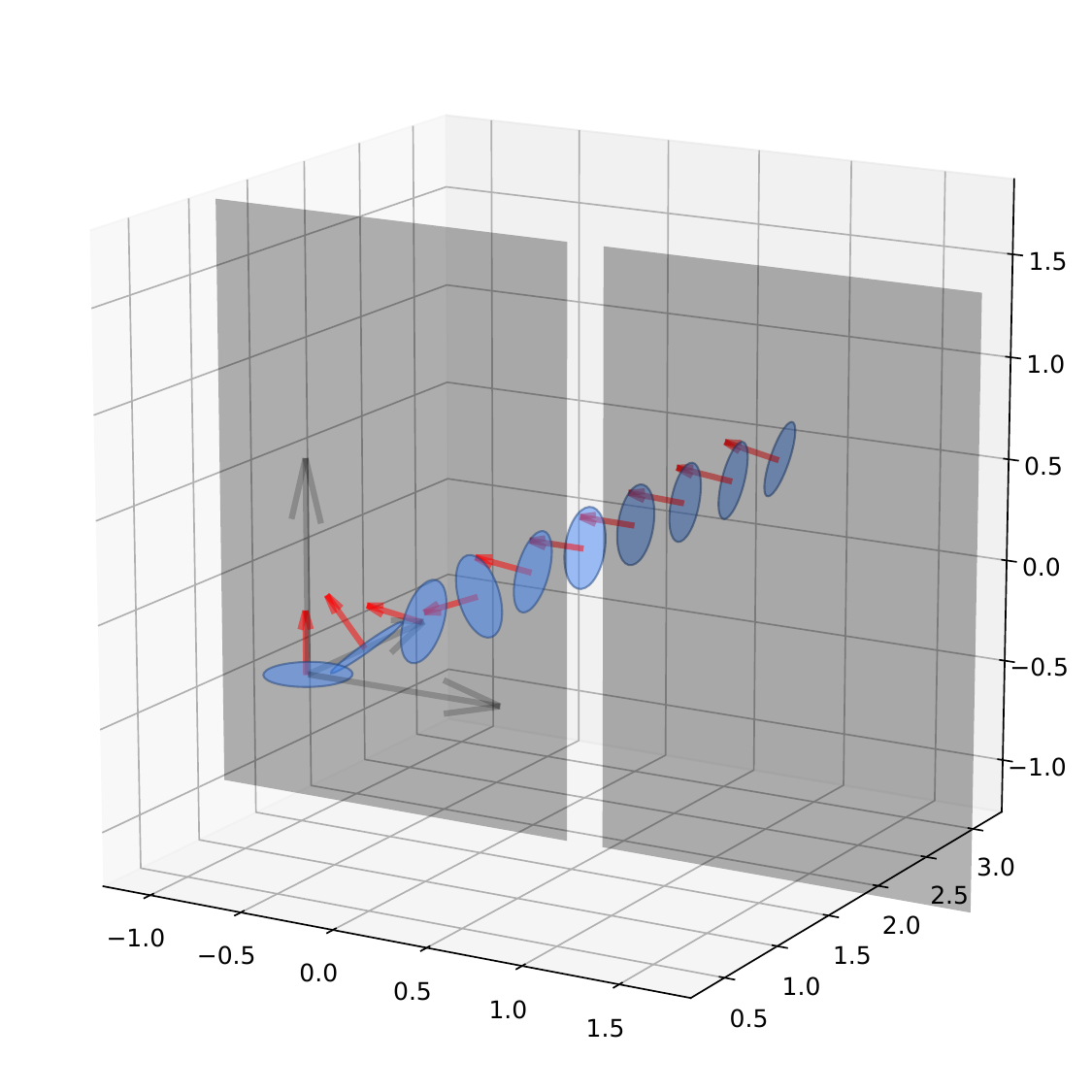}
            \label{fig:mpc/3d}
        } \hspace{-7mm} 
        \subfigure{
            \includegraphics[width=0.46\textwidth]{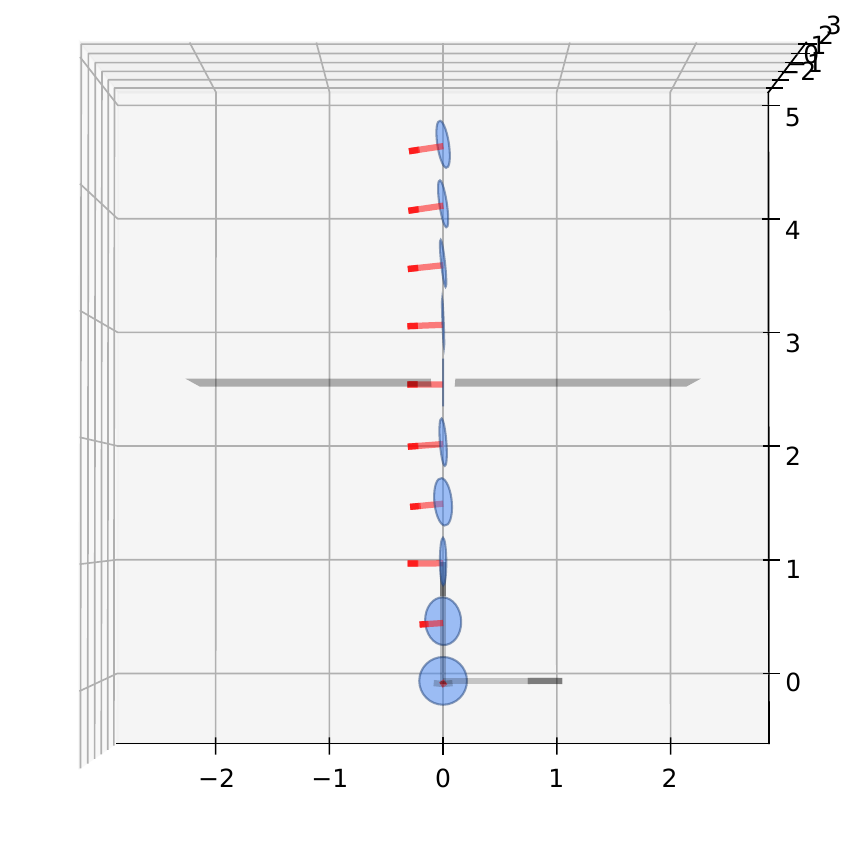}
            \label{fig:mpc/3dtop}
        } \\ \vspace{-2mm} 
        \subfigure{
            \includegraphics[width=0.9\textwidth]{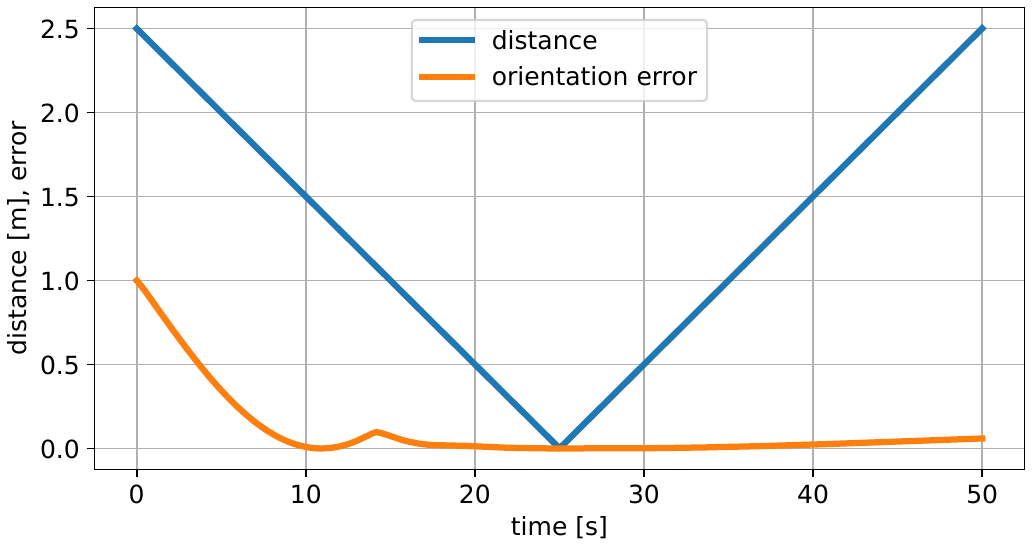}
            \label{fig:mpc/errors}
        }
        \caption{Lossless optimal control without torque control ($\vect{\tau}'=\vect{0}$)}
        \label{fig:mpc}
    \end{minipage}
\end{figure*}

\begin{figure*}[t]  % * makes the figure span across two columns
    \centering
    % Left column figure (Nominal attitude control)
    \begin{minipage}[t]{0.48\textwidth}  % Adjust width for left figure
        \centering
        \subfigure{
            \includegraphics[width=0.5\textwidth]{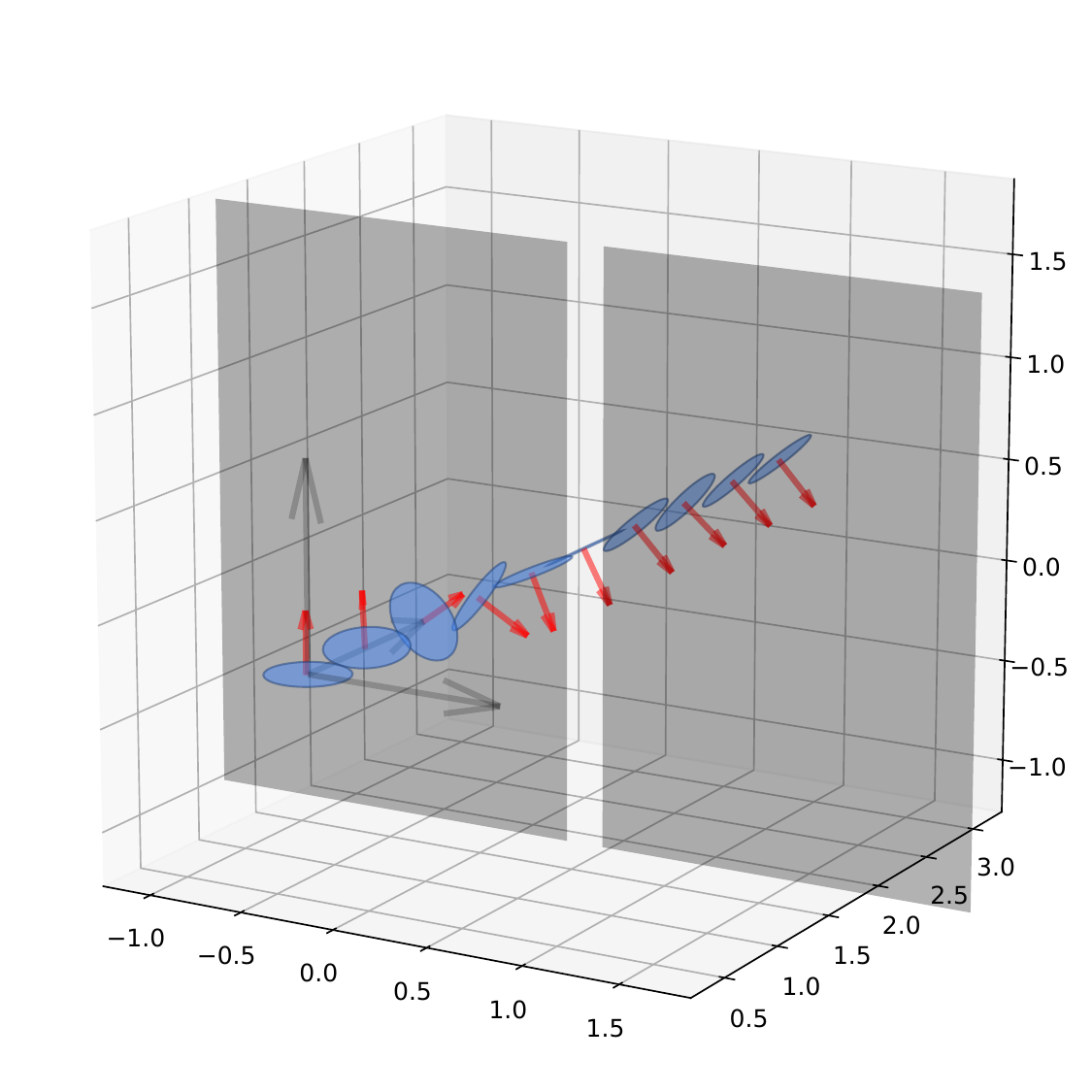}
            \label{fig:pd/3d}
        } \hspace{-7mm} 
        \subfigure{
            \includegraphics[width=0.46\textwidth]{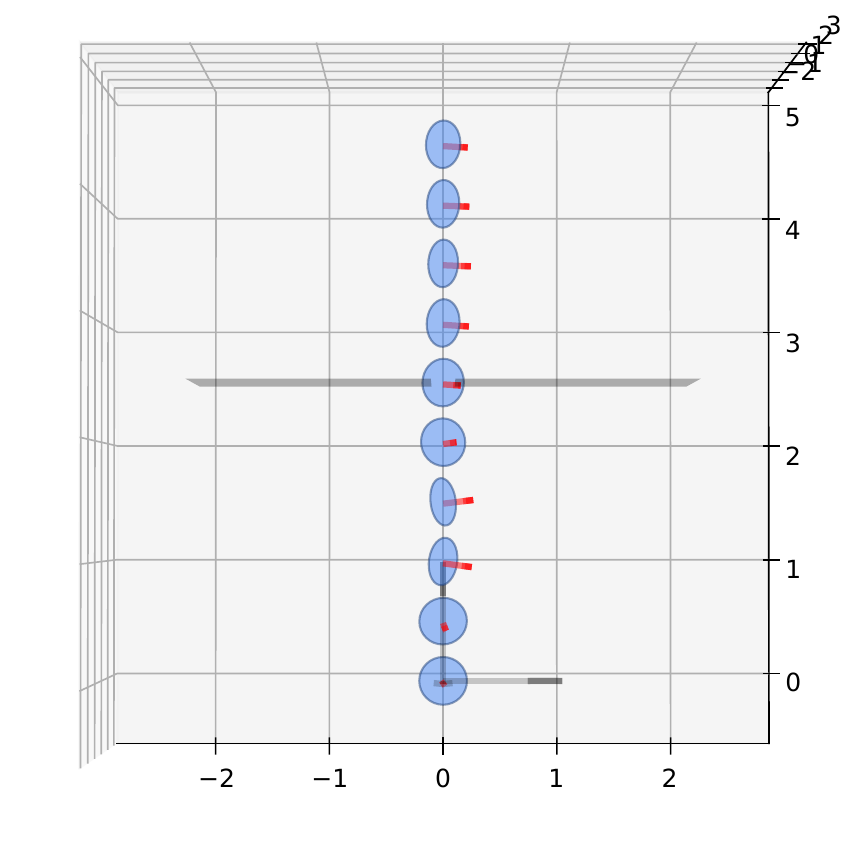}
            \label{fig:pd/3dtop}
        } \\ \vspace{-2mm} 
        \subfigure{
            \includegraphics[width=0.9\textwidth]{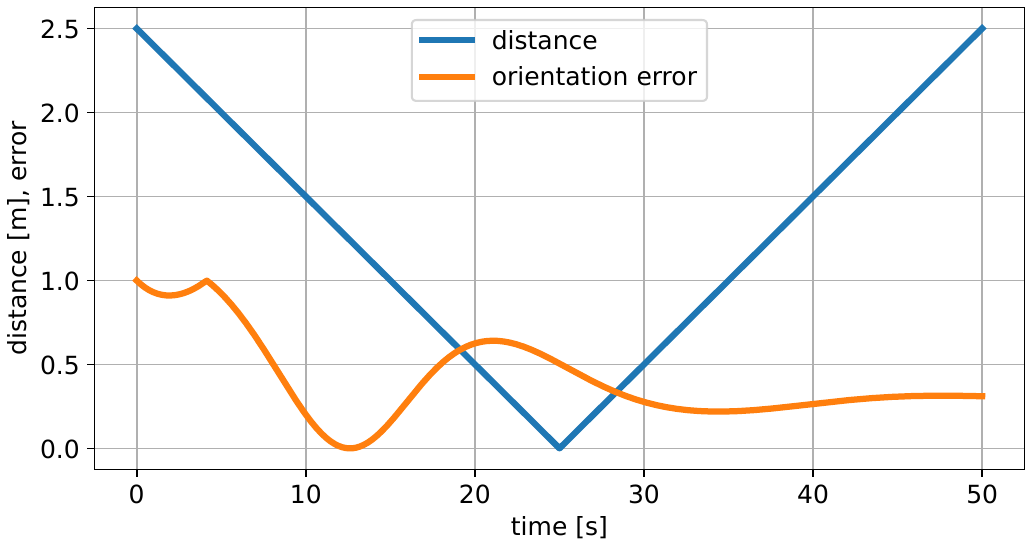}
            \label{fig:pd/errors}
        } \\ \vspace{-2mm} 
        \subfigure{
            \includegraphics[width=0.9\textwidth]{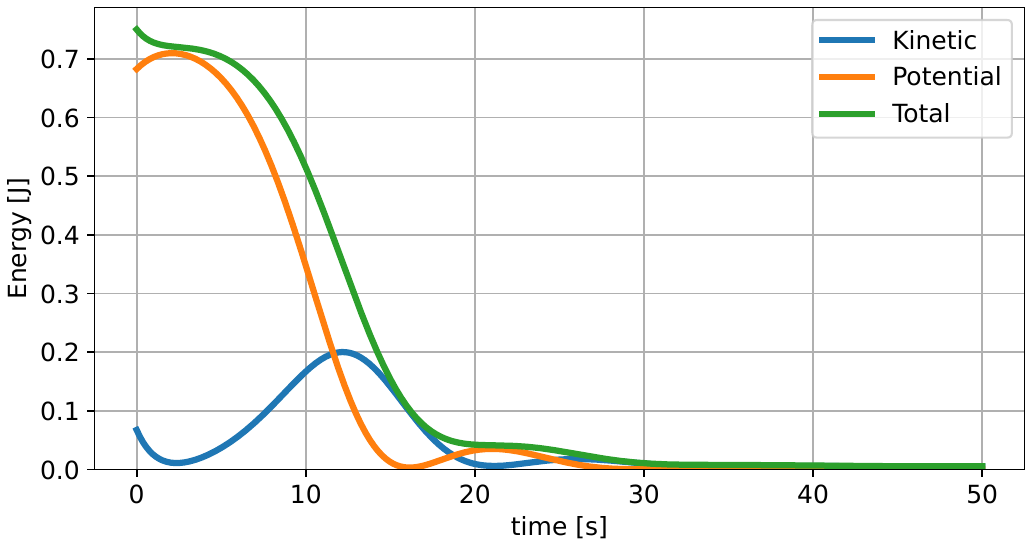}
            \label{fig:pd/energy}
        }
        \caption{Nominal attitude control alone ($\vect{a}=\vect{0}$)}
        \label{fig:pd}
    \end{minipage}%
    \hspace{5mm}  % Space between the two columns
    % Right column figure (Lossless optimal transient control)
    \begin{minipage}[t]{0.48\textwidth}  % Adjust width for right figure
        \centering
        \subfigure{
            \includegraphics[width=0.5\textwidth]{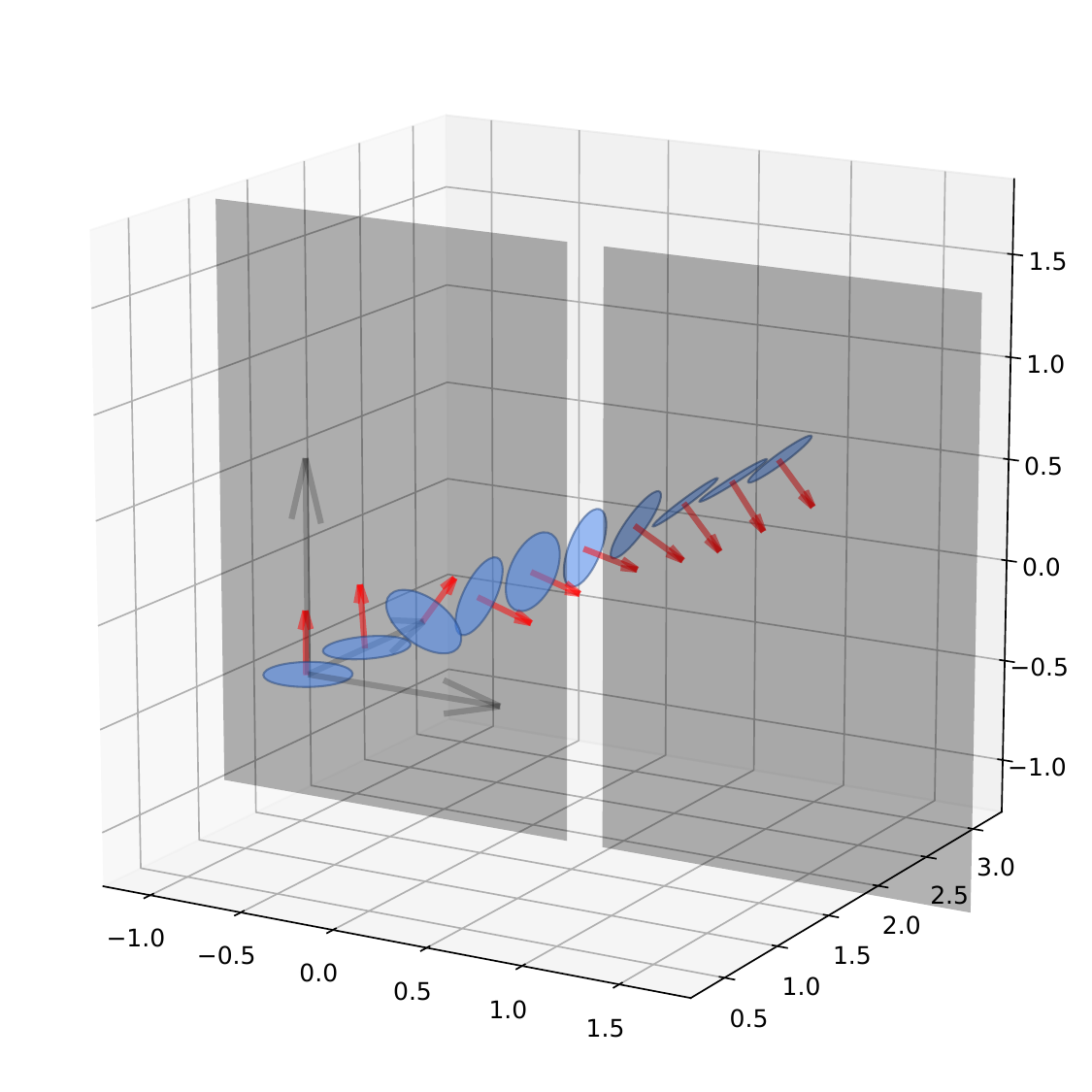}
            \label{fig:mpc_pd/3d}
        } \hspace{-7mm} 
        \subfigure{
            \includegraphics[width=0.46\textwidth]{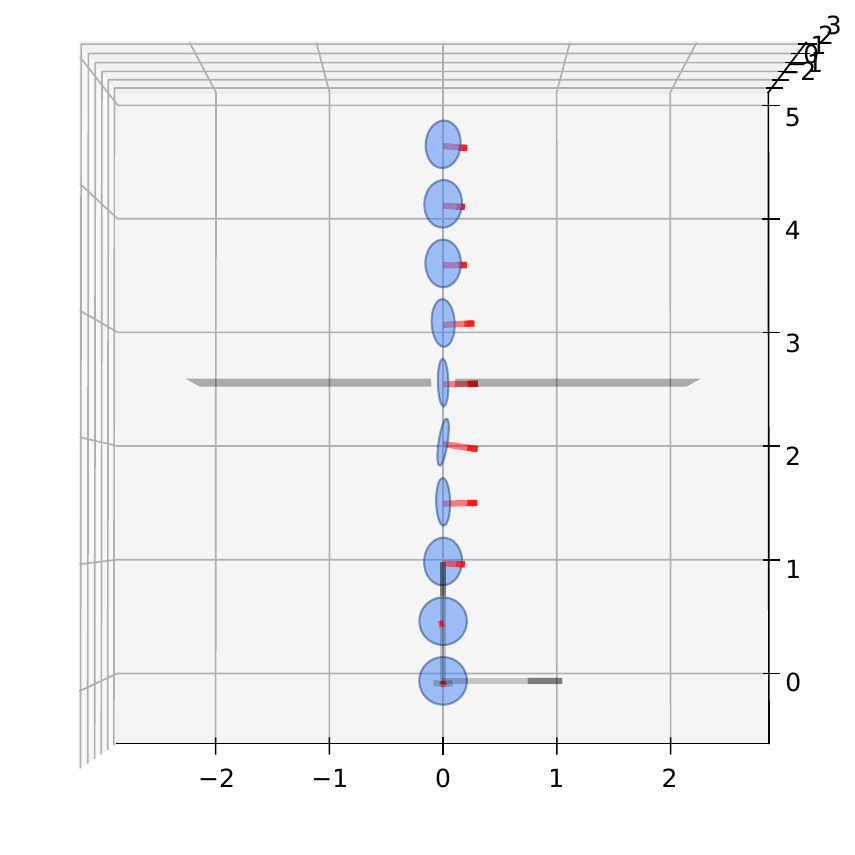}
            \label{fig:mpc_pd/3dtop}
        } \\ \vspace{-2mm} 
        \subfigure{
            \includegraphics[width=0.9\textwidth]{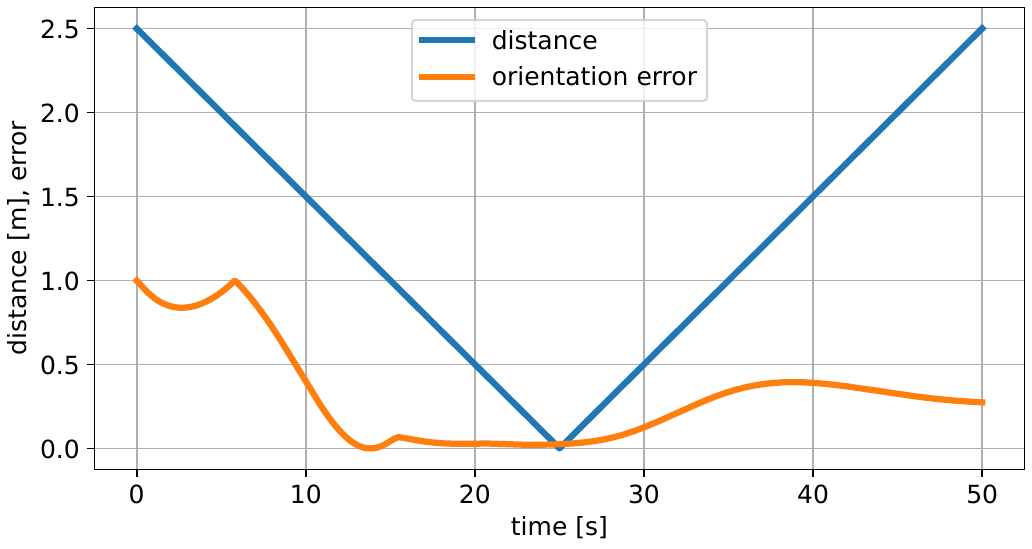}
            \label{fig:mpc_pd/errors}
        } \\ \vspace{-2mm} 
        \subfigure{
            \includegraphics[width=0.9\textwidth]{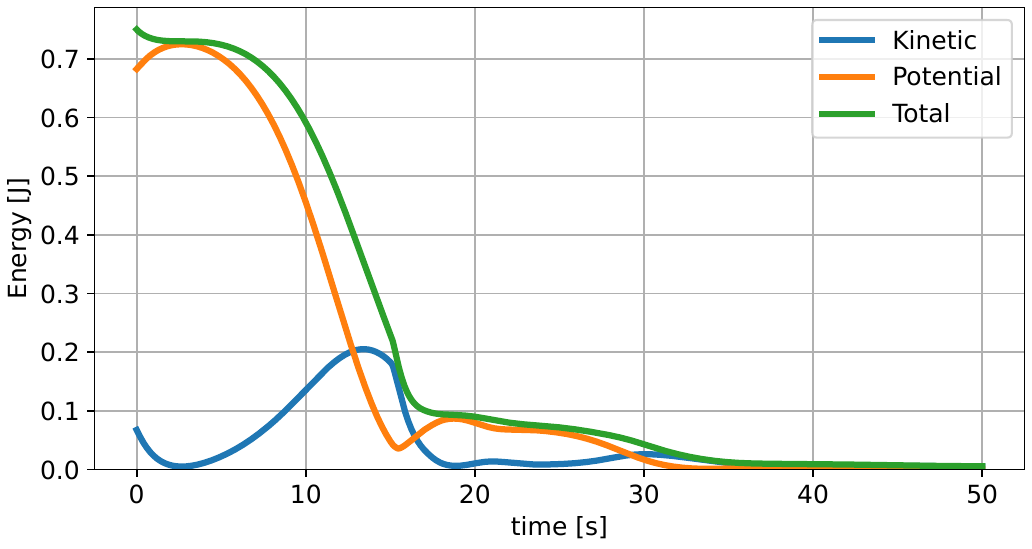}
            \label{fig:mpc_pd/energy}
        }
        \caption{Lossless optimal transient control and nominal attitude control}
        \label{fig:mpc_pd}
    \end{minipage}
\end{figure*}
 
\subsubsection{Results and discussion}
We analyse the effect of our solution under two conditions: with and without nominal torque control $\vect{\tau}'$. In the first condition, where $\vect{\tau}' = 0$, we compare the effects of lossless optimal control, reported in Fig.\ref{fig:mpc}, with the free dynamics of the rigid body (i.e., $\vect{a} = \vect{0}$), illustrated in Fig.\ref{fig:free}. In the second condition, the nominal torque control $\vect{\tau}'$ is applied using the passive controller defined in (\ref{eq:rot_ctrl}), which asymptotically stabilizes the system around the desired orientation equilibrium. In our simulations, we consider a rotation of $3\pi/4$ about the $y$-axis of the inertia frame. The effects of this torque control alone (with $\vect{a} = \vect{0}$) are shown in Fig. \ref{fig:pd}, while the combined effect of lossless optimal transient control and torque control is depicted in Fig. \ref{fig:mpc_pd}. For each simulation a 3D graphical representation of the rigid body dynamics is shown in the corresponding figure. The rigid body is depicted as a blue disc on the $x$ and $y$ axes of the body-fixed frame, with a red arrow along the $z$-axis. In the second plot of each figure, we report the distance between the rigid body and the slit $d_{p,p_\star}$ (in blue) and the orientation error  $\varepsilon_{R,R_\star}$ (in orange) over time. Ideally, for the rigid body to pass through the slit, the orientation error should be zero when the distance is at its minimum. %\markup{A yellow area close to zero show maximum orientation error $\varepsilon_{R,R_\star}<\todo{0.13}$ under which the system can pass through the slit derived using basic trigonometry from the slit width ($ 0.2 \si{m}$) and the disc diameter ($0.4 \si{m}$).} 
The bottom plots in Fig. \ref{fig:pd} and Fig. \ref{fig:mpc_pd} show the evolution in time of the kinetic, potential, and total energy.

The simulation results show a successful validation of the proposed controller: in both cases the activation of the NMPC routine is able to steer the rigid body along the slit without compromising the stability property of the closed-loop system. Particularly interesting is the comparison of the energy functions in Fig. \ref{fig:pd} and \ref{fig:mpc_pd}: the evolution of the closed-loop Lyapunov function (\ref{eq:Lyapunov}) is in both cases non increasing (i.e., it is a valid Lyapunov function for both cases, as predicted by Proposition \ref{prop:2} ) but the way it decreases (the transient behavior) is different. In the case powered by the NMPC the system trades kinetic energy for potential energy faster, in order to allow the rigid body to find the correct orientation to pass through the slit at the right time.

\section{Conclusions}
This work introduced a novel optimisation framework based on energetic arguments which is able to modify the transient behavior of controlled rigid body without compromising the passivity/stability properties of the closed-loop system.
Future work is hinged upon addressing practical challenges in real-time implementation and extend the framework to more complex environments.
\label{sec:conc}

\renewcommand*{\bibfont}{\small}
\printbibliography

\end{document}